\documentclass[prodmode]{acmsmall-hck}
\usepackage{hyperref}
\usepackage{amsmath}
\usepackage{epsfig}
\usepackage{algorithm,algorithmic}

\newcommand{\bmat}{\left[ \begin{array}}
\newcommand{\emat}{\end{array} \right]}

\newtheorem{claim}[theorem]{Claim}
\newtheorem{fact}[theorem]{Fact}
\newtheorem{comment}[theorem]{Comment}

\newcommand{\remove}[1]{}

\newcommand{\chold}{Cholesky decomposition}

\newcommand{\lt}{\left}
\newcommand{\rt}{\right}

\newcommand{\io}{I/O-complexity\,}

\newcommand{\ignore}[1]{}

\acmYear{2011}
\acmMonth{7}

\begin{document}

\markboth{Ballard, Demmel, Holtz, and Schwartz}
{Graph Expansion and Communication Costs of Fast Matrix Multiplication}

\title{Graph Expansion and Communication Costs of \\
Fast Matrix Multiplication}

\author{
Grey Ballard
\affil{University of California at Berkeley}
James Demmel
\affil{University of California at Berkeley}
Olga Holtz
\affil{University of California at Berkeley and Technische Universit\"at Berlin}
Oded Schwartz
\affil{University of California at Berkeley}
}

\acmformat{Ballard, G., Demmel, J., Holtz, O., and Schwartz, O. 2011.
Graph Expansion and Communication Costs of Fast Matrix Multiplication.}

\begin{bottomstuff}
A preliminary version of this paper appeared in Proceedings of the
23rd ACM Symposium on Parallelism in Algorithms and Architectures (SPAA'11) \cite{BallardDemmelHoltzSchwartz11b} and received the best paper award.

This work is supported by Microsoft (Award $\#$024263)
and Intel (Award $\#$024894) funding
and by matching funding by U.C. Discovery (Award $\#$DIG07-10227);
Additional support from Par Lab affiliates National Instruments, NEC,
Nokia, NVIDIA, and Samsung.
Research supported by U.S. Department of Energy grants under
Grant Numbers DE-SC0003959, DE-SC0004938, and DE-FC02-06-ER25786,
as well as Lawrence Berkeley National Laboratory
Contract DE-AC02-05CH11231; Research supported by the Sofja Kovalevskaja
programme of Alexander von Humboldt Foundation
and by the National Science Foundation under agreement DMS-0635607;
Research supported by ERC Starting Grant Number 239985.

Author's addresses: G.~Ballard, Computer Science Division,
University of California, Berkeley, CA 94720;
J.~Demmel, Mathematics Department and Computer Science Division,
University of California, Berkeley, CA 94720;
O.~Holtz, Department of Mathematics,
University of California, Berkeley;
O.~Schwartz, Computer Science Division,
University of California, Berkeley, CA 94720.
\end{bottomstuff}

\begin{abstract}
The communication cost of algorithms (also known as I/O-complexity) is shown to be closely
related to the expansion properties of the corresponding
computation  graphs. We demonstrate this on Strassen's and other
fast matrix multiplication algorithms, and obtain first lower
bounds on their communication costs.

In the sequential case, where the processor has a fast memory of size $M$,
too small to store three $n$-by-$n$ matrices,
the lower bound on
the number of words moved between fast and slow memory is,
for many of the matrix multiplication algorithms,
$\Omega\lt( \lt(\frac{n}{\sqrt M } \rt)^{\omega_0}\cdot M \right)$,
where $\omega_0$ is the exponent in the arithmetic count (e.g., $\omega_0 = \lg 7$ for Strassen, and $\omega_0 = 3$ for conventional
matrix multiplication). With $p$ parallel processors,
each with fast memory of size $M$, the lower bound is $p$ times smaller.

These bounds are attainable both for sequential and for parallel algorithms and hence optimal.
These bounds can also be attained by many fast algorithms in linear algebra
(e.g., algorithms for LU, QR, and solving the Sylvester equation).

\end{abstract}




\maketitle

\section{Introduction}
The communication of an algorithm (e.g., transferring data between
the CPU and memory devices, or between parallel processors, a.k.a.
I/O-complexity) often costs significantly more time than its
arithmetic.
It is therefore of
interest (1) to obtain lower bounds for the
communication needed, and (2) to
design and implement algorithms attaining these lower bounds.
Communication also requires much more energy
than arithmetic, and saving energy may be even more important than
saving time.

Communication time varies by orders of
magnitude, from $O\lt(10^{-9}\rt)$ second for an L1 cache
reference, to $O\lt(10^{-2}\rt)$ second for disk access. The variation can be
even more dramatic when communication occurs over networks or the
internet.
While Moore's Law predicts an exponential increase of hardware density in
general, the annual improvement rate of
time-per-arithmetic-operation has, over the years, consistently
exceeded  that of time-per-word read/write
\cite{GrahamSnirPatterson04,FullerMillett11}.  The fraction of running time spent
on communication is thus expected to increase further.

\subsection{Communication model}

We model communication costs of sequential and parallel architecture as follows. In the
sequential case, with two levels of memory hierarchy (fast and
slow), communication means reading data items ({\em words}) from
slow memory (of unbounded size), to fast memory (of size $M$) and writing data from fast memory to
slow memory\footnote{See \cite{BallardDemmelHoltzSchwartz10a} for definition of
a model with memory hierarchy,
and a reduction from the two-level model. All bounds in this paper thus apply to
the model with memory hierarchy as well.}.
Words that are stored contiguously in slow memory can be read or
written in a bundle which we will call a {\em message}. We assume
that a message of $n$ words can be communicated between fast and
slow memory in time $\alpha+\beta n$ where $\alpha$ is the {\em
latency} (seconds per message) and $\beta$ is the {\em inverse
bandwidth} (seconds per word). We define the {\em bandwidth cost} of
an algorithm to be the total number of words communicated and the {\em latency cost} of an algorithm to be the total number of messages communicated.
We assume that the input matrices initially reside in slow memory, and are too large to fit
in the smaller fast memory. Our goal then is to minimize both bandwidth and latency costs.\footnote{The sequential
communication model used here is sometimes called the
\emph{two-level I/O model} or \emph{disk access machine (DAM)}
model (see \cite{AggarwalVitter88,BenderBrodalFagerbergJacobVicari07,ChowdhuryRamachandran06}).  Our bandwidth cost model follows that of
\cite{HongKung81} and \cite{IronyToledoTiskin04} in that it
assumes the block-transfer size is one word of data ($B=1$ in the
common notation).  However, our model allows message sizes to vary from one word up to the maximum number of words that can fit in fast memory.}

In the parallel case, we assume $p$ processors, each with memory of size $M$ (or with larger memory size, as long as we never use more than $M$ in each processor).  We are interested in the communication among
the processors.  As in the sequential case, we assume that a message
of $n$ consecutively stored words can be communicated in time
$\alpha+\beta n$. This cost includes the time required to ``pack'' non-contiguous words into a single message, if necessary.
We assume that the input is initially evenly
distributed among all processors, so $M\cdot p$ is at least as large as the input.
Again, the bandwidth cost and latency cost are the word and message counts respectively. However, we count the number of words and messages
communicated along the critical path as defined in \cite{YangMiller88}
(i.e., two words that are communicated simultaneously are counted only once),
as this metric is closely related to the total running time of the algorithm.
As before, our goal is to minimize the number of words and messages
communicated.

We assume that
(1) the cost per flop is the same on each processor and the communication costs
($\alpha$ and $\beta$)
are the same between each pair of processors
(this assumption is for ease of presentation and can be dropped, using \cite{BallardDemmelGearhart11};
see Section~\ref{sec:other-hardware}),
(2) all communication is ``blocking'': a processor can send/receive a single message at a time, and cannot communicate and compute a flop simultaneously (the latter assumption can be dropped, affecting the running time by a factor of two at most),
and
(3) there is no communication resource contention
among processors. For example, if processor 0 sends a message of size $n$ to processor 1 at time 0,
and processor 2 sends a message of size $n$ to processor 3 also at time 0, the cost along the critical
path is $\alpha + \beta n$.  However, if both processor 0 and processor 1 try to send a message to
processor 2 at the same time, the cost along the critical path will be the sum of the costs of each message.

\subsection{The Computation Graph and Implementations of an Algorithm}

The computation performed by an algorithm on a given input
can be modeled (see Section \ref{sec:expansion}) as a computation directed acyclic graph {\em (CDAG) }:
We have a vertex for each input /
intermediate / output argument, and edges according to
direct dependencies (e.g., for the binary arithmetic operation $x:=y+z$
we have a directed edge from $v_y$ to $v_x$ and from $v_z$ to $v_x$,
where the vertices $v_x,v_y,v_z$ stand for the arguments $x,y,z$, respectively).

An implementation of an algorithm determines,
in the parallel model, which arithmetic operations are performed by
which of the $p$ processors.
This corresponds to partitioning the corresponding CDAG into $p$ parts.
Edges crossing between the various parts correspond to arguments that are in the possession of one processor,
but are needed by another processor, therefore relate to communication.
In the sequential model, an implementation determines the order of the
arithmetic operations, in a way that respects the partial ordering of the CDAG
(see Section \ref{sec:expansion} relating this to communication cost).

Implementations of an algorithm may vary greatly in their communication costs.
The {\em \io of an algorithm} is the minimum bandwidth cost of the algorithm,
over all possible implementations. The {\io of a problem} is defined to be the minimum \io of
all algorithms for this problem. A lower bound of the \io of an algorithm is therefore a results of the form: any implementation of algorithm $Alg$ requires at least $X$ communication. An upper bound is of the form: there is an implementation for algorithm $Alg$ that requires at most $X$ communication.
We detail below some of the \io lower and upper bounds of specific algorithms, or a class of algorithms. \io lower bounds for a {\em problem} are claims
 of the form: any algorithm for a problem $P$ requires at least $X$ communication. These are much harder to find
 (but see for example [Demmel, Grigori, Hoemmen, and Langou, 2008]).

The lower bounds in this paper are for all {\em implementations} for a family of algorithms: ``Strassen-like'' fast matrix multiplication. Generally speaking, a ``Strassen-like'' algorithm utilizes an algorithm for multiplying two constant-size matrices in order to recursively multiply matrices of arbitrary size;
see Section \ref{sec:other-mm} for precise definition and technical assumptions.

\subsection{Previous Work}
Consider the classical $\Theta(n^3)$
algorithm for matrix multiplication\footnote{By which we mean any algorithm that computes
using the $n^3$ multiplications, whether this is done recursively, iteratively, block-wise or any other way.}. While na\"{i}ve
implementations are communication inefficient,
communication-minimizing sequential and parallel variants of this
algorithm were constructed, and proved optimal, by matching lower
bounds
\cite{Cannon69,HongKung81,FrigoLeisersonProkopRamachandran99,IronyToledoTiskin04}.

In
\cite{BallardDemmelHoltzSchwartz10a,BallardDemmelHoltzSchwartz11a}
we generalize the results of \cite{HongKung81,IronyToledoTiskin04}
regarding matrix multiplication, to obtain new \io lower bounds
for a much wider variety of algorithms. Most of our bounds are
shown to be tight. This includes all ``classical'' algorithms for $LU$
factorization, Cholesky factorization, $LDL^T$ factorization, and many for the $QR$
factorization, and eigenvalues and singular
values algorithms. Thus we essentially cover all direct methods of linear
algebra. The results hold for dense matrix algorithms (most of
them have $O(n^3)$ complexity), as well as sparse matrix algorithms
 (whose running time depends on the number of non-zero
elements, and their locations). They apply to sequential and parallel algorithms, to
compositions of linear algebra operations (like computing the
powers of a matrix), and to certain graph-theoretic
problems\footnote{See \cite{MichaelPennerPrasanna02} for bounds on
graph-related problems, and our
\cite{BallardDemmelHoltzSchwartz11a} for a detailed list of
previously known and recently designed sequential and parallel
algorithms that attain the above mentioned lower bounds.}.

The optimal
algorithms for square matrix multiplication are well known.
Optimal algorithms for dense LU, Cholesky, QR,
eigenvalue problems and the SVD are more recent. These include
[\citeNP{Gustavson97};
\citeNP{Toledo97};
\citeNP{ElmrothGustavson98};
\citeNP{FrigoLeisersonProkopRamachandran99};
\citeNP{AhmedPingali00};
\citeNP{FrensWise03};
[Demmel, Grigori, Hoemmen, and Langou, 2008];
[Demmel, Grigori, and Xiang, 2008];
\citeNP{BallardDemmelHoltzSchwartz09a};
\citeNP{DavidDemmelGrigoriPeyronnet10};
\citeNP{DemmelGrigoriXiang10};
\citeNP{BallardDemmelDumitriu10}],
and are not part of
standard libraries like LAPACK \cite{LAPACK} and ScaLAPACK
\cite{SCALAPACK}. See \cite{BallardDemmelHoltzSchwartz11a} for more details.

In
\cite{BallardDemmelHoltzSchwartz10a,BallardDemmelHoltzSchwartz11a}
we use the approach of \cite{IronyToledoTiskin04}, based on the
Loomis-Whitney geometric theorem
\cite{LoomisWhitney49,BuragoZalgaller88}, by embedding segments of
the computation process into a three-dimensional cube. This
approach, however, is not suitable when distributivity is used, as
is the case in Strassen \cite{Strassen69} and other fast
matrix multiplication algorithms (e.g.,
\cite{CoppersmithWinograd90,CohnKleinbergSzegedyUmans05}).

While the \io of classic matrix multiplication and algorithms with
similar structure is quite well understood, this is not the case
for algorithms of more complex structure. The problem of minimizing
communication in parallel classical matrix multiplication was
addressed \cite{Cannon69} almost simultaneously with the
publication of Strassen's fast matrix multiplication
\cite{Strassen69}. Moreover, an \io lower bound for the classical
matrix multiplication algorithm has been known for three decades \cite{HongKung81}.
Nevertheless, the \io of Strassen's fast matrix multiplication and
similar algorithms has not been resolved.

In this paper we obtain first communication cost lower bounds for
Strassen's and other fast matrix multiplication algorithms, in the sequential and parallel models.
These bounds are attainable both for sequential and for parallel algorithms and so optimal.

\subsection{Communication Costs of Fast  Matrix Multiplication}
\subsubsection{Upper bound}
The \io $IO(n)$ of Strassen's algorithm
(see Algorithm \ref{alg:strassen}, Appendix \ref{sec:strassen-cited}),
applied to $n$-by-$n$ matrices on a machine with fast memory of
size $M$, can be bounded above as follows (for actual uses of
Strassen's algorithm, see \cite{DouglasHerouxSlishmanSmith94,Huss-LedermanJacobsonJohnsonTsaoTurnbull96,DesprezSuter04}):
Run the recursion until the matrices are sufficiently small. Then,
read the two input sub-matrices into the fast memory, perform the
matrix multiplication inside the fast memory, and write the
result into the slow memory\footnote{Here we assume that the recursion
tree is traversed in the usual depth-first order.}. We thus have $ IO(n) \leq 7 \cdot  IO
\lt(\frac{n}{2}\rt) + O(n^2)$ and $ IO\lt(\frac{\sqrt M}{3} \rt) =
O(M) $. Thus
\begin{equation}\label{eqn:strassen-upper}
IO(n) = O\lt( \lt(\frac{n}{\sqrt M } \rt)^{\lg 7}\cdot M \rt) .
\end{equation}

\subsubsection{Lower bound}

In this paper, we obtain a tight lower bound:
\begin{theorem}{\sc (Main Theorem)}\label{thm:main}
The \io $IO(n)$ of Strassen's algorithm on a machine with fast
memory of size $M$, assuming that no
intermediate values are computed twice\footnote{We assume no recomputation throughout the paper.}, is
\begin{equation}
IO(n) = \Omega\lt( \lt(\frac{n}{\sqrt M } \rt)^{\lg 7}\cdot M \rt)
.
\end{equation}
\end{theorem}
It holds for any implementation and any known variant of
Strassen's algorithm\footnote{This lower bound for the sequential case seems to contradict the upper bound from FOCS'99
\cite{FrigoLeisersonProkopRamachandran99,BlellochChowdhuryGibbonsRamachandranChenKozuch08}), due to a miscalculation (see \cite{Leiserson08} for details).
}$^{,}$\footnote{To obtain the lower bounds for latency costs we divide the bandwidth costs by the maximal message length, $M$. This holds for all the lower bounds here, both in the sequential and parallel models.}.
This includes Winograd's $O(n^{\lg 7})$
variant that uses 15 additions instead of 18, which is the most
used fast matrix multiplication algorithm in practice
\cite{DouglasHerouxSlishmanSmith94,Huss-LedermanJacobsonJohnsonTsaoTurnbull96,DesprezSuter04}.

For parallel algorithms, using a reduction from the sequential to the parallel
model (see e.g., \cite{IronyToledoTiskin04} or our
\cite{BallardDemmelHoltzSchwartz11a}) this yields:

\begin{corollary}\label{cor:parallel}
Let $IO(n)$ be the \io of Strassen's algorithm, run on a machine
with $p$ processors, each with a local memory of size $M$. Assume that no intermediate values are computed twice.
Then
\begin{eqnarray*}
IO(n) &=& \Omega\lt( \lt(\frac{n}{\sqrt M } \rt)^{\lg 7}\cdot
\frac{M}{p} \rt) .
\end{eqnarray*}
\end{corollary}

We can extend these bounds to a wider class of all ``Strassen-like''
fast matrix multiplication
algorithms.
Note that this class does not include all fast matrix multiplication algorithms
(see Section~\ref{sec:strassen-like-algs} for definition of ``Strassen-like'' algorithms,
and in particular the technical assumption in Section \ref{sec:technical-assumption}).
Let $Alg$ be any ``Strassen-like'' matrix multiplication algorithm
that runs in time $O(n^{\omega_0})$ for some $2<\omega_0<3$. Then, using the same arguments that lead to (\ref{eqn:strassen-upper}), the \io of $Alg$ can be shown to be
$
IO(n) = O\lt( \lt(\frac{n}{\sqrt M } \rt)^{\omega_0}\cdot M \rt)
$. We obtain a matching lower bound:
\begin{theorem}\label{thm:general}
The \io $IO(n)$ of a recursive ``Strassen-like'' fast matrix multiplication algorithm
 with $O(n^{\omega_0})$
arithmetic operations, on a machine with fast memory of size $M$
is
\begin{equation}
IO(n) = \Omega\lt( \lt(\frac{n}{\sqrt M } \rt)^{\omega_0}\cdot M
\rt) .
\end{equation}
\end{theorem}

Note that or the cubic recursive algorithm for matrix multiplication,
$\omega_0= \lg 8 = 3$, and the above formula is
$IO(n) = \Omega\lt( \lt(\frac{n}{\sqrt M } \rt)^{3}\cdot M \rt) =
\Omega\lt( \frac{n^3}{\sqrt M } \rt)$
and identifies with the lower bounds of \cite{HongKung81}
and \cite{IronyToledoTiskin04}.
While the lower bounds for $\omega_0=3$ and for $\omega_0 < 3$ have the same form, the proofs are completely
different, and it is not clear whether our approach can be used to prove their lower bounds and vice versa.

\begin{corollary}\label{cor:parallel-str-like}
Let $IO(n)$ be the \io of a ``Strassen-like'' algorithm
(with arithmetic performed as in Theorem \ref{thm:general}), run on a machine with
$p$ processors, each with a local memory of size $M$. Assume that no intermediate values are computed
twice.
Then
\begin{eqnarray*}
IO(n) &=& \Omega\lt( \lt(\frac{n}{\sqrt M } \rt)^{\omega_0}\cdot
\frac{M}{p} \rt).
\end{eqnarray*}
\end{corollary}

\subsection{The Expansion Approach}

The proof of the main theorem is based on estimating the edge
expansion of the computation directed acyclic graph {\em (CDAG) }
of an algorithm. The \io is shown to be closely related to the edge expansion
properties of this graph. As the graph has a recursive structure,
the expansion can be analyzed directly (combinatorially, similarly
to what is done in \cite{Mihail89,AlonSchwartzShapira08,KouckyKabanetsKolokolova10}) or by spectral
analysis (in the spirit of what was done for the Zig-Zag expanders
\cite{ReingoldVadhanWigderson00}). There is, however, a new
technical challenge. The replacement product and the Zig-Zag product
act similarly on all vertices. This is not what happens
in our case: multiplication and addition vertices behave differently.

The expansion approach is similar to the one taken by Hong and Kung
\cite{HongKung81}. They use the red-blue pebble game to obtain
tight lower bounds on the \io of many algorithms, including
classical $\Theta(n^3)$ matrix multiplication, matrix-vector multiplication, and
FFT. The proof is obtained by showing that the size of any
subset of the vertices of the CDAG is bounded by a function of
the size of its dominator set (recall that a dominator set $D$ for $S$
is a set of vertices such that every path from
an input vertex to a vertex in $S$  contains some vertex in $D$).

On the one hand, their dominator set technique has the advantage of
allowing recomputation of any intermediate value.
We were not able to allow recomputation using our edge expansion approach.
On the other hand, the dominator set requires large input or output.
Such an assumption is not needed by the edge expansion approach,
as the bounds are guaranteed
by edge expansion of many (internal) parts of the CDAG.
In
that regard, one can view the approach of
\cite{IronyToledoTiskin04} (also in
[Demmel, Grigori, Hoemmen, and Langou, 2008;
\citeNP{BallardDemmelHoltzSchwartz10a};
\citeNP{BallardDemmelHoltzSchwartz11a}])
as an edge expansion assertion on the CDAGs of the corresponding
classical algorithms.

The study of expansion properties of a CDAG was also suggested as
one of the main motivations of Lev and Valiant \cite{LevValiant83}
in their work on superconcentrators. They point out many papers
proving that classes of algorithms computing DFT, matrix inversion
and other problems all have to have CDAGs with good expansion properties,
thus providing lower bounds on the number of the arithmetic
operations required.

Other papers study connections between bounded space
computation, and combinatorial expansion-related properties of the
corresponding CDAG (see e.g.,
\cite{Savage94,BilardiPreparata99,BilardiPietracaprinaD'Alberto00}
and references therein).

\subsection{Paper organization}
Section \ref{sec:preliminaries} contains preliminaries on the notions
of graph expansion. In Section
\ref{sec:expansion} we state and prove the connection between \io
and the expansion properties of the computation graph. In Section
\ref{sec:strassen} we analyze the expansion of the CDAG of Strassen's
algorithm.
We discuss the generalization
of the bounds to other algorithms in Section \ref{sec:other-mm}, and present conclusions and open
problems in Section~\ref{sec:open}.

\section{Preliminaries}\label{sec:preliminaries}

\subsubsection{Edge expansion}
The edge expansion $h(G)$ of a $d$-regular undirected
graph $G=(V,E)$ is:
\begin{equation}
h(G) \equiv \min_{U \subseteq V, |U| \leq |V|/2} \frac{|E(U, V
\setminus U )|}{d \cdot |U|}
\end{equation}
where $E(A,B)\equiv E_G(A,B)$ is the set of edges connecting the vertex
sets $A$ and $B$. We omit the subscript $G$ when
the context makes it clear.

\subsubsection{When $G$ is not regular}
Note that CDAGs are typically not regular.
If a graph $G=(V,E)$ is not regular but has a bounded maximal degree $d$,
then we can add ($<d$) loops to vertices of degree $<d$, obtaining
a regular graph $G'$. We use the convention that a loop adds 1 to the degree of a vertex.
Note that for any $S \subseteq V$,
we have $|E_{G}(S,V \setminus S)| = |E_{G'}(S,V \setminus S)|$, as
none of the added loops contributes to the edge expansion
of~$G'$.

\subsubsection{Expansion of small sets}
For many graphs, small sets expand better than larger sets. Let $h_s(G)$ denote the edge expansion for
sets of size at most $s$ in $G$:
\begin{equation}
h_s(G) \equiv \min_{U \subseteq V, |U| \leq s} \frac{|E(U, V
\setminus U )|}{d \cdot |U|} ~.
\end{equation}
In many cases, $h_s(G)$ does not depend on $|V(G)|$, although it
may decrease when $s$ increases. One way of bounding $h_s(G)$ is
by decomposing $G$ into small subgraphs of large edge expansion.

\begin{claim}\label{clm:small-sets}
Let $G=(V,E)$ be a $d$-regular graph that can be decomposed into
edge-disjoint (but not necessarily vertex-disjoint) copies of a
$d'$-regular graph $G'=(V',E')$. Then the edge expansion
of $G$ for sets of size at most $|V'|/2$ is $h(G')\cdot
\frac{d'}{d}$, namely
\begin{equation*}
h_{\frac{|V'|}{2}}(G) \equiv \min_{U \subseteq V, |U| \leq |V'|/2}
\frac{|E_{G}(U, V \setminus U )|}{d \cdot|U|} \geq h(G')\cdot
\frac{d'}{d} ~.
\end{equation*}
\end{claim}
For proving this claim, recall the definition of graph decomposition:
\begin{definition}[Graph decomposition]
We say that the set of graphs $\{G'_i = (V_i,E_i) \}_{i \in [l]}$
is an {\em edge-disjoint decomposition} of $G=(V,E)$ if
$V=\bigcup_i V_i$ and $E=\biguplus_i E_i$.
\end{definition}

\begin{proof}(of Claim \ref{clm:small-sets})
Let $U \subseteq V$ be of size $U \leq |V'|/2$. Let $\{G'_i =
(V_i,E_i) \}_{i \in [l]}$ be an edge-disjoint decomposition of
$G$, where every  $G_i$ is isomorphic to $G'$. Then
\begin{eqnarray*}
|E_G(U,V \setminus U)| &=& \sum_{i \in [l]} |E_{G'_i}(U_i, V_i
\setminus U_i)|
 \geq  \sum_{i \in [l]} h(G'_i)\cdot d' \cdot |U_i|
\\ &=&  h(G') \cdot d' \cdot \sum_{i \in [l]} |U_i|
 \geq  h(G') \cdot d' \cdot |U|~.
\end{eqnarray*}
Therefore
$
\frac{|E_G(U,V \setminus U)|}{d \cdot |U|}  \geq  h(G') \cdot \frac{d'}{d}~.
$
\end{proof}

\section{I/O-Complexity and  Edge Expansion}\label{sec:expansion}
In this section we recall the notion of computation graph of an algorithm,
then show how a partition argument connects the expansion
properties of the computation graph and the \io of the algorithm. A similar
partition argument already appeared in \cite{IronyToledoTiskin04},
and then in our \cite{BallardDemmelHoltzSchwartz11a}. In both
cases it is used to relate \io to the Loomis-Whitney geometric
bound \cite{LoomisWhitney49}, which can be viewed, in this
context, as an expansion guarantee for the corresponding graphs.

\subsection{The computation graph}

For a given algorithm, we consider the computation (directed)
graph $G=(V,E)$, where there is a vertex for each arithmetic
operation {\em (AO)} performed, and for every input element. $G$
contains a directed edge $(u,v)$, if the output operand of the AO
corresponding to $u$ (or the input element corresponding to $u$),
is an input operand to the AO corresponding to $v$. The in-degree
of any vertex of $G$ is, therefore, at most 2 (as the arithmetic operations are binary).
The out-degree is,
in general, unbounded\footnote{As the lower bounds
are derived for the bounded out-degree case, we will show how to convert the corresponding CDAG to obtain constant out-degree, without affecting the
 \io too much.}, i.e., it may be a function of $|V|$.  We
next show how an expansion analysis of this graph can be used to
obtain the \io lower bound for the corresponding algorithm.

\subsection{The partition argument}

Let $M$ be the size of the fast memory. Let $O$ be any total
ordering of the vertices that respects the partial ordering of the
CDAG $G$, i.e., all the edges are going up in the total order. This total
ordering can be thought of as the actual order in which the
computations are performed. Let $P$ be any partition of $V$ into
segments $S_1,S_2,...$, so that a segment $S_i \in P$ is a subset of
the vertices that are contiguous in the total ordering $O$.

Let $R_S$ and $W_S$ be the set of read and write operands,
respectively (see Figure \ref{fig:RW}). Namely, $R_S$ is the set
of vertices outside $S$ that have an edge going into $S$, and
$W_S$ is the set of vertices in $S$ that have an edge going
outside of $S$. Then the total \io due to reads of AOs in $S$ is
at least $|R_S|-M$, as at most $M$ of the needed $|R_S|$ operands
are already in fast memory when the execution of the segment's AOs
starts. Similarly, $S$ causes at least $|W_S|-M$ actual write
operations, as at most $M$ of the operands needed by other
segments are left in the fast memory when the execution of the
segment's AOs ends. The total \io is therefore bounded below
by\footnote{One can think of this as a game: the first player
orders the vertices. The second player partitions them into
contiguous segments. The objective of the first player (e.g., a
good programmer) is to order the vertices so that any consecutive
partitioning by the second player leads to a small communication
count.}

\begin{eqnarray}\label{eqn:RW} IO &\geq& \max_{P} \sum_{S \in P}
\lt( |R_S| + |W_S| - 2M \rt)~.
\end{eqnarray}

\begin{figure}[h]
\begin{center}
\psfig{file=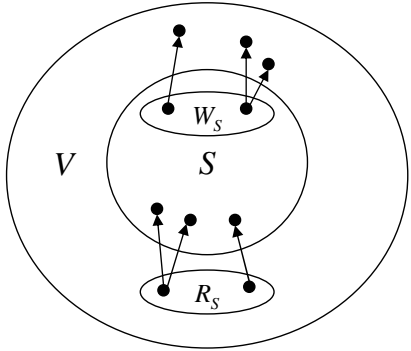 ,clip=, scale=0.7} \protect\caption{A subset
(segment) $S$ and its corresponding read operands $R_S$, and write
operands $W_S$. }\label{fig:RW}
\end{center}
\end{figure}

\subsection{Edge expansion and \io}

Consider a segment $S$ and its read and write operands $R_S$ and
$W_S$ (see Figure \ref{fig:RW}). If the graph $G$ containing $S$
has $h(G)$ edge expansion\footnote{The direction of the edges does
not matter much for the expansion-bandwidth argument: treating all
edges as undirected changes the \io estimate by a factor of 2 at
most. For simplicity, we will treat $G$ as undirected.}, maximum
degree $d$ and at least $2|S|$ vertices, then (using the definition
of $h(G)$), we have
\begin{claim}\label{clm:RW}
$|R_S|+|W_S| \geq \frac12 \cdot h(G) \cdot |S|$~.
\end{claim}
\begin{proof}
We have $|E(S, V \setminus S)| \geq h(G) \cdot d \cdot |S|$.
Either (at least) half of the edges $E(S, V \setminus S)$ touch
$R_S$ or half of them touch $W_S$. As every vertex is of degree
$d$, we have $|R_S|+|W_S| \geq \max\{|R_S|,|W_S|\} \geq \frac 1d
\cdot \frac12 \cdot |E(S, V \setminus S)| \geq h(G) \cdot |S| /2$.
\end{proof}

\noindent
Combining this with (\ref{eqn:RW}) and choosing to partition
$V$ into $|V|/s$ segments of equal size $s$, we obtain:
$
IO \geq  \max_{s}\frac{|V|}{s} \cdot \lt( \frac{h(G) \cdot s}{2}
- 2M \rt) = \Omega\lt(|V|\cdot h(G) \rt)
$. In many cases $h(G)$ is too small to attain the desired \io lower
bound. Typically, $h(G)$ is a decreasing function in $|V(G)|$,
namely the edge expansion deteriorates with the increase of the
input size and with the running time of the corresponding
algorithm. This is the case with matrix multiplication algorithms:
the cubic, as well as the Strassen and ``Strassen-like'' algorithms.
In such cases, it is better to consider the expansion of $G$ on
small sets only:
$
IO \geq  \max_{s}\frac{|V|}{s} \cdot \lt( \frac{h_s(G) \cdot
s}{2} - 2M \rt)
$.
Choosing\footnote{The existence of a value $s$ that satisfies the
condition is not always guaranteed. In the next section we confirm
this for Strassen, for sufficiently large $|V(G)|$ (in particular,
$|V(G)|$ has to be larger than $M$). Indeed this is the
interesting case, as otherwise all computations can be performed
inside the fast memory, with no communication, except for reading
the input once.} the minimal $s$ so that
\begin{eqnarray}\label{eqn:IO-general-expansion-condition}
\frac{h_s(G) \cdot s}{2} \geq 3M
\end{eqnarray} we obtain
\begin{eqnarray}\label{eqn:IO-general-expansion}
IO &\geq&  \frac{|V|}{s} \cdot M~.
\end{eqnarray}

In some cases, the computation graph $G$ does not fit this
analysis: it may not be regular, it may have vertices of unbounded
degree, or its edge expansion may be hard to analyze. In such
cases, we may consider some subgraph $G'$ of $G$ instead to
obtain a lower bound on the \io:
\begin{claim}\label{clm:subgraph}
Let $G=(V,E)$ be a computation graph of an algorithm $Alg$. Let
$G'=(V',E')$ be a subgraph of $G$, i.e., $V' \subseteq V$
and $E'  \subseteq E$. If $G'$ is $d$-regular and $\alpha =
\frac{|V'|}{|V|}$, then the \io of $Alg$ is

\begin{equation}\label{eqn:IO-expansion-subgraph}
IO \geq \frac{\alpha}{2} \cdot \frac{|V|}{s} \cdot M
\end{equation}
where $s$ is chosen so that ${h_s(G') \cdot \alpha s \over 2} \geq 3M~
$.
\end{claim}

\noindent The correctness of this claim follows from Equations
(\ref{eqn:IO-general-expansion-condition}) and
(\ref{eqn:IO-general-expansion}), and from the fact that at least
an $\alpha/2$ fraction of the segments have at least $\frac{\alpha}{2} \cdot
s$ of their vertices in $G'$ (otherwise $V' <
\frac{\alpha}{2}\cdot V/s \cdot s + (1 - \frac{\alpha}{2})\cdot V/s \cdot \frac{\alpha}{2} s
< \alpha V
$). We therefore have:
\begin{lemma}\label{lem:central}
Let $Alg$ be an algorithm with $AO(N)$ arithmetic
operations ($N$ being the total input size, $N=\Theta(n^2)$ for matrix multiplication) and computation graph
$G(N)=(V,E)$. Let $G'(N)=(V',E')$ be a regular constant degree
subgraph of $G$, with $\frac{|V'|}{|V|} = \Theta(1)$. Then the \io
of $Alg$\footnote{In Strassen's algorithm, $N=2n^2$ is the number
of input matrices elements and $T(N) = \Theta\lt(n^{\omega_0}\rt) =
\Theta\lt(N^{\omega_0/2} \rt)$. $G'$ is the graph $Dec_kC$ for $k =
\lg M$, see Section \ref{sec:strassen} for the definition of
$Dec_kC$.} on a machine with fast memory of size $M$ is
\begin{equation}\label{eqn:central-basic}
IO = \Omega\lt(|V'| \cdot h_s(G'(N)) \rt) \;\; \text{ for} \;\; s
= AO(M) ~.
\end{equation}
As $AO(N) = \Theta(|V'|)$ and $h_s(G'(N))$ for $s
= AO(M)$ is $\Theta(h(G'(M)))$ (recall Claim \ref{clm:small-sets}) we obtain, equivalently,
\begin{equation}\label{eqn:central-nice}
IO = \Omega\lt(AO(N) \cdot h(G'(M)) \rt)~.
\end{equation}

\end{lemma}

\section{Expansion Properties of Strassen's Algorithm}\label{sec:strassen}

Recall Strassen's algorithm for matrix multiplication (see
Algorithm \ref{alg:strassen} in  Appendix  \ref{sec:strassen-cited})
and consider its computation graph (see Figure \ref{fig:comp-graph}).
Let $H_{i}$ be computation graph of Strassen's algorithm for recursion of depth $i$, hence
$H_{\lg n}$ corresponds to the computation for
input matrices of size $n \times n$. $H_{\lg n}$  has the following
structure:
\begin{itemize}
\item
Encode $A$: generate weighted sums of elements of $A$
(this corresponds to the left factors of lines 5-11 of the algorithm).
\item
Similarly encode $B$ (this corresponds to the right factors of lines 5-11 of the algorithm).
\item
Then multiply the encodings of $A$ and $B$
element-wise (this corresponds to line 2 of the algorithm).
\item
Finally, decode $C$, by taking weighted sums of the
products (this corresponds to lines 12-15 of the algorithm).
\end{itemize}

\begin{figure}[th!]
\begin{center}
\psfig{file=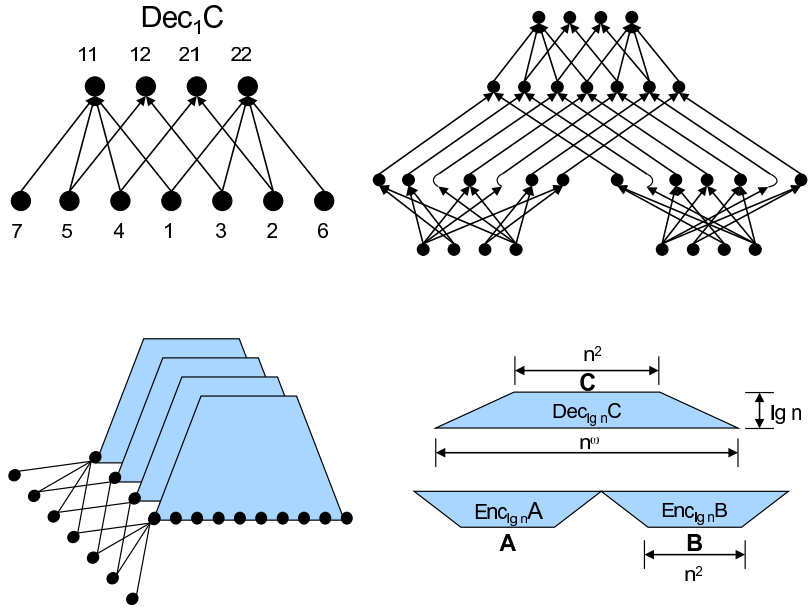 ,clip=, scale=0.6} \protect\caption{The
computation graph of Strassen's algorithm (See Algorithm
\ref{alg:strassen} in Appendix).
\newline Top left:     $Dec_1C$. Top right:    $H_1$.
Bottom left:  $Dec_{\lg n}C$. Bottom right: $H_{\lg n}$.
}\label{fig:comp-graph}
\end{center}
\end{figure}

\begin{comment}\label{cmt:constant-in-degree}
$Dec_{1}C$ is presented, for simplicity, with vertices of in-degree larger than two (but constant).
A vertex of degree larger than two, in fact, represents a full binary (not necessarily balanced) tree.
Note that replacing these high in-degree vertices
with trees changes the edge expansion of the graph by a constant factor at most
(as this graph is of constant size, and connected).
Moreover, there is no change in the number of input and output vertices.
Therefore the arguments in the following proof of Lemma \ref{lem:Dec-expansion} still hold.
\end{comment}

\subsection{The computation graph for $n$-by-$n$ matrices}
Assume w.l.o.g.~that $n$ is an integer power of $2$. Denote by
$Enc_{\lg n}A$ the part of $H_{\lg n}$ that corresponds to the encoding
of matrix $A$. Similarly, $Enc_{\lg n}B$, and $Dec_{\lg n}C$
correspond to the parts of $H_{\lg n}$ that compute the encoding
of $B$ and the decoding of $C$, respectively.

\subsubsection{A top-down construction of the computation graph}
We next construct the computation graph $H_{i+1}$ by constructing $Dec_{i+1}C$ (from $Dec_i C$ and $Dec_1 C$)
and similarly constructing $Enc_{i+1}A$ and $Enc_{i+1}B$, then composing the three parts together.
\begin{itemize}
\item
Duplicate $Dec_{1}C$ $7^i$ times.
\item
Duplicate $Dec_{i}C$ four times.
\item
Identify the $4 \cdot 7^i$ output vertices of the copies of
$Dec_{1}C$ with the $4 \cdot 7^i$ input vertices of the copies of $Dec_{i}C$:
\begin{itemize}
\item
Recall that each $Dec_{1}C$ has four output vertices.
\item
The first output vertex of the $7^i$
$Dec_{1}C$ graphs are identified with the $7^i$ input vertices of the first copy of $Dec_{i}C$.
\item
The second output vertex of the $7^i$
$Dec_{1}C$ graphs are identified with the $7^i$ input vertices of the second copy of $Dec_{i}C$.
And so on.
\item
We make sure that the $j$th input vertex of a copy of $Dec_{i}C$ is identified
with an output vertex of the $j$th copy of $Dec_{1}C$.
\end{itemize}
\item
We similarly obtain $Enc_{i+1}A$ from $Enc_{i}A$ and $Enc_{1}A$,
\item
and $Enc_{i+1}B$ from $Enc_{i}B$ and $Enc_{1}B$.
\item
For every $i$, $H_{i}$ is obtained by connecting edges from the $j$th output vertices of $Enc_{i}A$ and $Enc_{i}B$
to the $j$th input vertex of $Dec_{i}C$.
\end{itemize}
This completes the construction. Let us note some properties of these graphs.

The graph $Dec_1C$ has no vertices which are both input
and output. As all out-degrees are at most 4 and all in degree are at most 2 (Recall Comment \ref{cmt:constant-in-degree}) we have:
\begin{fact}\label{fct:constant-degree}
All vertices of $Dec_{\lg n}C$ are of degree at most $6$.
\end{fact}

 However, $Enc_1A$
and $Enc_1B$ have vertices which are both input and output (e.g.,
$A_{11}$), therefore $Enc_{\lg n}A$ and $Enc_{\lg n}B$ have vertices
of out-degree $\Theta(\lg n)$. All in-degrees are at most $2$, as
an arithmetic operation has at most two inputs.

As $H_{\lg n}$ contains vertices of large degrees, it is easier to
consider  $Dec_{\lg n}C$: it contains only vertices of constant
bounded degree,  yet at least one third of the vertices of $H_{\lg
n}$ are in it.

\begin{lemma}{\sc (Main lemma)}\label{lem:Dec-expansion}
The edge expansion of $Dec_{k}C$ is $$h(Dec_{k}C)=\Omega\lt(\lt(\frac{4}{7}\rt)^{k}\rt)$$
\end{lemma}
\noindent The proof follows below, but first note that it suffices
to deduce the expansion of $Dec_{\lg n}C$ on small sets.
Assume w.l.o.g.~that $n$ is an integer power of
$\sqrt{M}$.\footnote{We may assume this, as we are dealing with a
lower bound here, so it suffices to prove the assertion for an
infinite number of $n$'s. Alternatively, in the following
decomposition argument, we leave out a few of the top or bottom
levels of vertices of $Dec_{\lg n}C$, so that $n$ is an integer
power of $\sqrt{M}$ and so that at most $|S|/2$ vertices of $S$
are cut off. } Then $Dec_{\lg n}C$ can be split into edge-disjoint copies of $Dec_{\frac12 \lg M  }C$. Using Claim
\ref{clm:small-sets}, we thus have:
\begin{corollary}
$s \cdot h_{s}(Dec_{\lg n}C) \geq 3M$ for $s = 9 \cdot M^{\lg 7 / 2}$.
\end{corollary}

\noindent  As $Dec_{\lg n}C$ contains $\alpha=\frac13$ of the vertices of
$H_{\lg n}$,
Lemma \ref{lem:central} now yields Main
Theorem \ref{thm:main}. Note that $Dec_{\lg n}C$ has no input
vertices, so no restriction on input replication is needed.

\subsubsection{Combinatorial  Estimation of the Expansion}\label{sec:expansion-combinatorial}
\begin{proof}[of Lemma \ref{lem:Dec-expansion}]
Let $G_k=(V,E)$ be $ Dec_kC$, and let $S \subseteq V, |S| \leq
|V|/2$. We next show that $|E(S,V \setminus S)| \geq c \cdot d
\cdot |S| \cdot \lt(\frac47 \rt) ^k $, where $c$ is some universal
constant, and $d$ is the constant degree of $Dec_kC$ (after
adding loops  to make it regular).

The proof works as follows. Recall that $G_k$ is a layered graph (with layers corresponding to recursion steps),
so all edges (excluding loops)
connect between consecutive levels of vertices. We argue (in Claim \ref{clm:S-homogeneity}) that each level of
$G_k$ contains about the same fraction of $S$ vertices, or else we have
many edges leaving $S$. We also observe (in Fact \ref{fct:heterogeneity}) that such
homogeneity (of a fraction of $S$ vertices) does not hold between
distinct parts of the lowest level, or, again, we have many
edges leaving $S$. We then show that the homogeneity between levels,
combined with the heterogeneity of the lowest level,
guarantees that there are many edges leaving $S$.

Let $l_i$ be the $i$th level of vertices of $G_k$, so $4^k = |l_1|
< |l_2| < \cdots < |l_i| = 4^{k-i+1}7^{i-1} < \cdots < |l_{k+1}| = 7^k$.
Let $S_i \equiv S \cap l_i$. Let $\sigma = \frac{|S|}{|V|}$ be the fractional size
of $S$ and $\sigma_i = \frac{|S_i| }{|l_i|}$ be the fractional size of $S$ at level $i$.
Due to averaging, we observe the following:
\begin{fact} There
exist $i$ and $i'$ such that $\sigma_i \leq \sigma \leq \sigma_{i'}$.
\end{fact}

\begin{fact}\label{fct:l1}
\begin{eqnarray*}
|V| &= &\sum _{i =1}^{k+1} |l_{i}|
=
 \sum_{i =1}^{k+1} |l_{k+1}| \cdot \lt(\frac{4}{7}\rt)^i\\
&=&  |l_{k+1}| \cdot \lt(1 - \lt(\frac{4}{7}\rt)^{k+2}\rt) \cdot \frac73\\
&=&  \left(\frac47 \right)^k\cdot |l_{1}| \cdot \lt(1 - \lt(\frac{4}{7}\rt)^{k+2}\rt) \cdot \frac73
\end{eqnarray*}
so
$ \frac37 \leq \frac{|l_{k+1}|}{|V|} \leq \frac37 \cdot \frac{1}{1- \lt(\frac{4}{7}\rt)^{k+2}  }
$, and
$
\frac37 \cdot \lt(\frac{4}{7}\rt)^{k}
\leq \frac{|l_1|}{|V|} \leq \frac37 \cdot \lt(\frac{4}{7}\rt)^{k} \cdot \frac{1}{1- \lt(\frac{4}{7}\rt)^{k+2}  }.
$
\end{fact}

\begin{claim}\label{clm:delta-cost}
There exists $c' = c'(G_1)$ so that $|E(S, V \setminus S) \cap E(l_i,l_{i+1})| \geq c'\cdot d \cdot|\delta_i| \cdot
|l_i| $.
\end{claim}
\begin{proof}
Let $G'$ be a $G_1$ component connecting $l_i$ with $l_{i+1}$ (so it has four vertices in $l_i$ and seven in $l_{i+1}$).
$G'$ has no edges in $E(S,V \setminus
S)$ if all or none of its vertices are in $S$. Otherwise, as $G'$ is connected, it
contributes at least one edge to $E(S,V \setminus S)$.
The number of such $G_1$ components with all their vertices in $S$
is at most $\min\{\sigma_i,\sigma_{i+1}\}\cdot \frac{|l_i|}{4}$. Therefore, there are at least
$|\sigma_{i}-\sigma_{i+1}|\cdot \frac{|l_i|}{4}$  $G_1$ components with at least one vertex in $S$ and one vertex that is not.
\end{proof}

\begin{claim}[Homogeneity between levels]\label{clm:S-homogeneity}
If there exists $i$ so that $\frac{|\sigma-\sigma_i|}{\sigma}\geq \frac1{10}$,  then
$$|E(S,V \setminus S)| \geq c \cdot d \cdot |S|\cdot \lt( \frac47\rt)^k $$
where $c>0$ is some constant depending on $G_1$ only.
\end{claim}
\begin{proof}
Assume that there exists $j$ so that $\frac{|\sigma-\sigma_j|}{\sigma}\geq \frac1{10}$.
Let $\delta_i \equiv \sigma_{i+1}-\sigma_i$.
By Claim \ref{clm:delta-cost}, we have
\begin{eqnarray*}
|E(S,V \setminus S)|
&\geq& \sum _{i \in [k]} |E(S, V \setminus S) \cap E(l_i,l_{i+1})|\\
&\geq& \sum _{i \in [k]} c'\cdot d\cdot|\delta_i| \cdot |l_i|\\
&\geq& c'\cdot d\cdot |l_1| \sum _{i \in [k]}   |\delta_i|\\
&\geq& c'\cdot d\cdot |l_1|
 \cdot \lt(\max_{i \in [k+1]} \sigma_i
- \min_{i \in [k+1]} \sigma_i\rt) .
\end{eqnarray*}
By the initial assumption, there exists $j$ so that $\frac{|\sigma-\sigma_j|}{\sigma}\geq \frac1{10}$, therefore $\max_i \sigma_i - \min_i \sigma_i \geq \frac{\sigma}{10}$, then
\begin{align*}
|E(S,V \setminus S)|
&\geq  c'\cdot d\cdot |l_1| \cdot \frac{\sigma}{10}\\
\intertext{By Fact \ref{fct:l1}, $|l_1| \geq  \frac37 \cdot \lt ( \frac{4}{7} \rt)^k \cdot |V| $, }
&\geq c'\cdot d \cdot \frac37 \cdot \lt( \frac{4}{7} \rt)^k \cdot |V| \cdot \frac{\sigma}{10}\\
\intertext{As $|S| = \sigma \cdot |V|$,}
&\geq c \cdot d\cdot |S| \cdot \lt( \frac47\rt)^k
\end{align*}
for any $c \leq \frac{c'}{10}\cdot \frac37$.
\end{proof}

Let $T_k$ be a tree corresponding to the recursive construction of $G_k$ in the following way (see Figure \ref{fig:tree2}):
$T_k$ is a tree of height $k+1$, where each internal node has four children.
The root $r$ of $T_k$ corresponds to $l_{k+1}$ (the largest level of $G_k$).
The four children of $r$ correspond to the largest levels of the four graphs that one can obtain by
removing the level of vertices $l_{k+1}$ from $G_k$. And so on.
For every node $u$ of $T_k$, denote by $V_u$ the set of vertices in $G_k$ corresponding to $u$.
We thus have $|V_r|=7^k$ where $r$ is the root of $T_k$,
$|V_u| = 7^{k-1}$ for each node $u$ that is a child of $r$;
and in general we have $4^{i}$ tree nodes $u$ corresponding to a set of size $|V_u| = 7^{k-i+1}$.
Each leaf $l$ corresponds to a set of size $1$.

\begin{figure}[h]
\begin{center}
\psfig{file=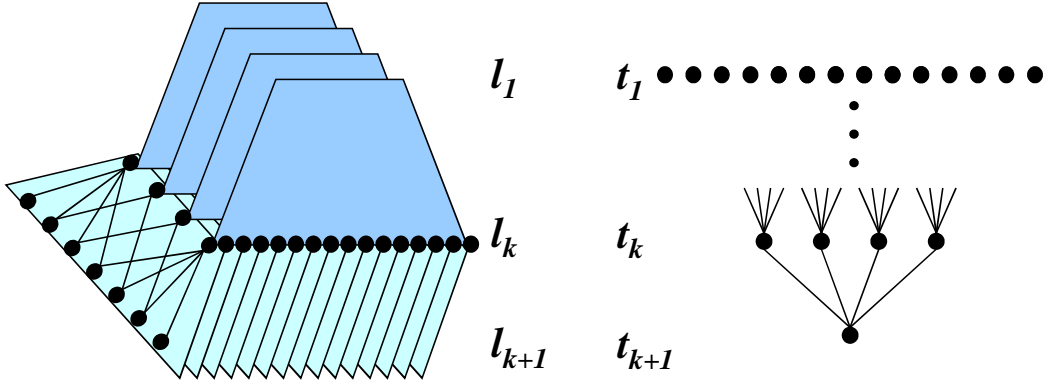 ,clip=, scale=0.8}
\protect\caption{The graph $G_k$ and its corresponding tree $T_k$. }
\label{fig:tree2}
\end{center}
\end{figure}

For a tree node $u$, let us define $\rho_u = \frac{|S \cap V_u|}{|V_u|}$ to be the fraction of $S$ nodes in $V_u$,
and $\delta_u = |\rho_u - \rho_{p(u)}|$, where $p(u)$ is the parent of $u$ (for the root $r$ we let $p(r)=r$).
We let $t_i$ be the $i$th level of $T_k$, counting from the bottom, so $t_{k+1}$ is the root and $t_{1}$ are the leaves.

\begin{fact}\label{fct:heterogeneity}
As $V_r=l_{k+1}$ we have $\rho_r = \sigma_{k+1}$. For a tree leaf $u \in t_1$, we have $|V_u|=1$. Therefore $\rho_u \in \{0,1\}$.
The number of vertices $u$ in $t_1$ with $\rho_u=1$ is $\sigma_1 \cdot |l_1|$.
\end{fact}

\begin{claim}\label{clm:delta-tree-cost}
Let $u_0$ be an internal tree node, and let $u_1,u_2,u_3,u_4$ be its four children. Then
$$\sum_i|E(S, V \setminus S) \cap E(V_{u_i},V_{u_0})| \geq c''\cdot d \cdot \sum_i  |\rho_{u_i}-\rho_{u_0}| \cdot |V_{u_i}|$$
where $c''=c''(G_1)$.
\end{claim}
\begin{proof}
The proof follows that of Claim \ref{clm:delta-cost}.
Let $G'$ be a $G_1$ component connecting $V_{u_0}$ with $\bigcup_{i \in [4]}V_{u_i}$ (so it has seven vertices in $V_{u_0}$
and one in each of $V_{u_1}$,$V_{u_2}$,$V_{u_3}$,$V_{u_4}$).
$G'$ has no edges in $E(S,V \setminus
S)$ if all or none of its vertices are in $S$. Otherwise, as $G'$ is connected, it
contributes at least one edge to $E(S,V \setminus S)$.
The number of $G_1$ components with all their vertices in $S$
is at most $\min\{\rho_{u_0},\rho_{u_1},\rho_{u_2},\rho_{u_3},\rho_{u_4}\}\cdot \frac{|V_{u_1}|}{4}$. Therefore, there are at least
$\max_{i \in [4]}\{|\rho_{u_0}-\rho_{u_i}|\}\cdot \frac{|V_{u_1}|}{4} \geq \frac{1}{16} \cdot \sum_{i\in [4]} |\rho_{u_i}-\rho_{u_0}| \cdot |V_{u_i}|$ \; $G_1$  components with at least one vertex in $S$ and one vertex that is not.
\end{proof}

\begin{align*}
\intertext{We have}
\nonumber
|E(S,V \setminus S)|
&= \sum _{u \in T_k} |E(S, V \setminus S) \cap E(V_u,V_{p(u)})|\\
\nonumber
\intertext{By Claim \ref{clm:delta-tree-cost}, this is}
\nonumber
&\geq  \sum _{u \in T_k} c''\cdot d\cdot  |\rho_{u}-\rho_{p(u)}|  \cdot |V_u|\\
&= c''\cdot d\cdot \sum_{i \in [k]} \sum _{u \in t_i} |\rho_{u}-\rho_{p(u)}|  \cdot 7^{i-1}\\
\nonumber
&\geq c''\cdot d\cdot \sum_{i \in [k]} \sum _{u \in t_i}  |\rho_u-\rho_{p(u)}| \cdot 4^{i-1} \\
\nonumber
\intertext{As each internal node has four children, this is}
&= c''\cdot d\cdot \sum _{v \in t_1} \sum_{u\in v \sim r}  |\rho_u-\rho_{p(u)}| \\
\nonumber
\intertext{where $v \sim r$ is the path from $v$ to the root $r$. By the triangle inequality for the function~$|~\cdot~|$}
&\geq c''\cdot d\cdot \sum _{v \in t_1}  |\rho_u-\rho_r|\\
\nonumber
\intertext{By Fact \ref{fct:heterogeneity}, }
&\geq  c''\cdot d\cdot  |l_1| \cdot ((1-\sigma_1) \cdot \rho_{r} + \sigma_1 \cdot (1-\rho_{r}) )\\
\intertext{By Claim \ref{clm:S-homogeneity}, w.l.o.g.,
$|\sigma_{k+1} - \sigma|/\sigma \leq \frac{1}{10}$ and
$|\sigma_1 - \sigma|/\sigma \leq \frac{1}{10}$. As $\rho_r = \sigma_{k+1}$,}
& \geq \frac34 \cdot c''\cdot d\cdot  |l_1| \cdot \sigma  \\
\intertext{and by Fact \ref{fct:l1},}
&\geq c\cdot d\cdot  |S| \cdot \left(\frac47 \right)^{k}\\
\intertext{for any $c \leq \frac34 \cdot c''$.}
\end{align*}
This completes the proof of Lemma \ref{lem:Dec-expansion}.
\end{proof}

\section{Other Algorithms}\label{sec:other-mm}
We now discuss the applicability of our approach to other algorithms, starting with other fast matrix multiplication algorithms.
\subsection{``Strassen-like'' Algorithms}\label{sec:strassen-like-algs}
A ``Strassen-like'' algorithm has
a recursive structure that utilizes a base case: multiplying two
$n_0$-by-$n_0$ matrices using $m(n_0)$ multiplications.
Given two matrices of size $n$-by-$n$, it splits them into $n_0^2$
blocks (each of size $\frac{n}{n_0}$-by-$\frac{n}{n_0}$), and works blockwise,
according to the base case algorithm. Additions (and subtractions) in the base case are interpreted as
additions (and subtractions) of blocks. These are performed element-wise. Multiplications in the base case are interpreted as
multiplications of blocks. These are performed by recursively calling the algorithm.
The
arithmetic count of the algorithm is then
$T(n) = m(n_0) \cdot T\lt( \frac{n}{n_0} \rt) + O(n^2)$, so $T(n)
= \Theta(n^{\omega_0})$ where $\omega_0 = \log_{n_0}m(n_0)$.

This is the structure of all the fast matrix
multiplication algorithms that were obtained since Strassen's
\cite{Pan80,Bini80,Schonhage81,Romani82,CoppersmithWinograd82,Strassen87,CoppersmithWinograd87},
(see \cite{BurgisserClausenShokrollahi97} for discussion of these algorithms), as well as \cite{CohnKleinbergSzegedyUmans05}, where the base case utilizes a novel group-theoretic approach.
In fact, any fast matrix multiplication algorithm can be
converted into this form \cite{Raz03},
and can even be made numerically stable while preserving this form
[Demmel, Dumitriu, Holtz, and Kleinberg, 2007].

\subsubsection{A critical technical assumption}\label{sec:technical-assumption}
For our technique to work, we further demand that the $Dec_1C$ part of the computation graph is a connected graph, in order to be  ``Strassen-like'' (this was assumed in the proof of Claim \ref{clm:delta-cost}).
Thus the ``Strassen-like'' class includes Winograd's variant of Strassen's algorithm \cite{Winograd71},
which uses 15 additions
rather than 18, but not the cubic algorithm, where
$Dec_1C$ is composed of four disconnected graphs (corresponding to the four outputs).
We conjecture that $Dec_1C$ is indeed connected
for all existing fast matrix-multiplication algorithms.
We note that the demand of connectivity of $Dec_1C$ may be waved in some cases (see \cite{BallardDemmelHoltzSchwartz11mix}).

\subsubsection{The communication costs of ``Strassen-like'' algorithms}

To prove Theorem \ref{thm:general}, which generalizes the \io lower bound of Strassen's algorithm (Theorem \ref{thm:main}) to all ``Strassen-like'' algorithms, we note the following:
The entire proof of Theorem \ref{thm:main}, and in particular,
the computations in the proof of Lemma \ref{lem:Dec-expansion},
hold for any ``Strassen-like'' algorithm, where we plug in
$n_0^2,m(n_0)$, and $\frac{n_0}{m(n_0)}$ instead of $4,7$, and
$\frac47$. For bounding the asymptotic \io, we do not care about
the number of internal vertices of $Dec_1C$; we need only to know
that $Dec_1C$ is connected (this critical technical assumption is used in the proof of Claim \ref{clm:delta-cost}), and to know the sizes $n_0$ and $m(n_0)$.
The only nontrivial adjustment is to show the equivalent of Fact \ref{fct:constant-degree}: that the graph $Dec_{\log
n}C$ is of bounded degree.

\begin{claim}\label{clm:constant-degree-general}
The $Dec_{\log n}C$ graph of any ``Strassen-like'' algorithm is of
degree bounded by a constant.
\end{claim}
\begin{proof}
If the set of input vertices of $Dec_1C$, and the set of its
output vertices are disjoint, then $Dec_{\log n}C$ is of constant
bounded degree (its maximal degree is at most twice the largest
degree of $Dec_1C$).

 Assume (towards contradiction) that the base graph $Dec_1C$ has
an input vertex which is also an output vertex. An output vertex
represents the inner product of two $n_0$-long vectors, i.e., the corresponding row-vector of $A$ and column vector of $B$. The corresponding bilinear polynomial is irreducible.
This is a
contradiction, since
an input
vertex represents the multiplication of a (weighted) sum of elements of
$A$ with a (weighted) sum of elements of $B$.
\end{proof}

\subsection{Uniform, Non-stationary Fast Matrix Multiplication Algorithms}

Another class of matrix multiplication algorithms, the {\em uniform, non-stationary} algorithms,
allows mixing of schemes of the previous (``Strassen-like'') class.
In each recursive level, a different scheme may be used. The
CDAG has a repeating structure inside one level, but the structure
may differ between two distinct levels. This class
includes, for example, algorithms that optimize for input sizes
(for sizes that are not an integer power of a constant integer).
The class also includes algorithms
that cut the recursion off at some point, and then switch to the classical
algorithm. For these and other implementation issues, see
\cite{DouglasHerouxSlishmanSmith94,Huss-LedermanJacobsonJohnsonTsaoTurnbull96}
(sequential model) and \cite{DesprezSuter04} (parallel model).
The \io lower bound generalizes to this class, and will appear in a separate note \cite{BallardDemmelHoltzSchwartz11mix}.

\subsection{Non-uniform, Non-stationary Fast Matrix Multiplication Algorithms}

A third class, the {\em non-uniform, non-stationary}
algorithms, allows recursive calls to have different structure, even
when they refer to multiplication of matrices in the same recursive level.
It is not clear
how to analyze the expansion of the CDAG of an algorithm in the third class, although we are not aware of any algorithms in this class.
Such an analysis, applied to the base case of \cite{CohnKleinbergSzegedyUmans05},
may improve the \io lower bound for fast matrix multiplication by a (large) constant.

\subsection{Multiplying Rectangular Matrices}

Multiplication of rectangular matrices have seen a series of increasingly fast algorithms
culminating in Coppersmith's algorithm \cite{Coppersmith97}.
It is possible to extend our approach and obtain the first lower bounds on the communication costs for these algorithms, and show that in some cases they are attainable, and therefore optimal \cite{BallardDemmelHoltzSchwartz11rect}.

\subsection{Other Algorithms}
Fast matrix multiplication algorithms are basic building blocks in many fast algorithms in
linear algebra, such as algorithms for LU, QR, and solving the Sylvester
equation [Demmel, Dumitriu, and Holtz, 2007]. Therefore, \io lower bounds  for these algorithms
can be derived from our lower bounds for fast matrix multiplication algorithms
\cite{BallardDemmelHoltzSchwartz11lin}.
For example, a lower bound on LU (or QR, etc.) follows when the
fast matrix multiplication algorithm is called by the LU algorithm on
sufficiently large subblocks of the matrix.
This is the case in the algorithms of [Demmel, Dumitriu, and Holtz, 2007],
 and we can then deduce matching lower and upper bounds \cite{BallardDemmelHoltzSchwartz11lin}.

\section{Conclusions and Open Problems}\label{sec:open}
We obtained a tight lower bound for the \io of Strassen's and
``Strassen-like'' fast matrix multiplication algorithms. These bounds are
optimal for the sequential model with two memory levels and with
memory hierarchy. The lower bounds extend to the parallel model and other models.
Recently these bounds were attained (up to an $O(\log p)$ factor) by new parallel implementations,
for Strassen's algorithm and for ``Strassen-like'' algorithms \cite{BallardDemmelHoltzRomSchwartz11}.

\subsection{Memory constraints for the classical and ``Strassen-like'' matrix multiplication algorithms}\label{sec:recursive-implementations}
Some (parallel) algorithms require very little,
up to a constant factor extra memory beyond what is necessary to keep
the input and output. These are sometimes called {\em linear space algorithms}.
One class of such algorithms are the ``2D'' algorithms for classical matrix multiplication,
that use two-dimensional grid of processors. Here we allow
$M = \Theta \lt( \frac{n^2}{p} \rt)$ local memory use (recall that $p$ is the number of processors,
and $n$ the dimension of the matrices),
thus no replication of the input matrices is allowed
\cite{Cannon69}.

If the underlying grid of $p$ processors is a three-dimensional mesh,
and the available memory per processor
is larger by a factor of $p^{\frac13}$ than the minimum necessary to store the input and output matrices,
then a ``3D'' algorithm can be used (see
\cite{DekelNassimiSahni81,AggarwalChandraSnir90,AgarwalBalleGustavsonJoshiPalkar95,McCollTiskin99}).
These ``3D'' algorithms can reduce the communication cost by a factor of $p^{1/6}$,
down to $\Theta \lt( \frac{n^2}{p^{\frac23}} \rt)$ \cite{AggarwalChandraSnir90}, attaining the lower bounds
\cite{IronyToledoTiskin04,BallardDemmelHoltzSchwartz11a} that take into account any
amount of replication.

Recently, Demmel and Solomonik \cite{SolomonikDemmel11} showed how to combine these two extremes
into one algorithm (named ``2.5D'') and obtained a communication efficient implementation for classical matrix multiplication, for
local memory size $M = \Theta \lt( c \cdot \frac{n^2}{p} \rt)$ for any $1 \leq c \leq p^{\frac13}$.
See Table~\ref{tbl:lower-upper}.

Using Corollaries \ref{cor:parallel} and \ref{cor:parallel-str-like}, and plugging in $M = \Theta \lt( \frac{n^2}{p} \rt)$, $M = \Theta \lt( \frac{n^2}{p^{\frac23}} \rt)$ , and $M = \Theta \lt( c \cdot \frac{n^2}{p} \rt)$ we obtain
corresponding lower bounds for ``Strassen-like'' algorithm with various restriction local memory sizes.
These were recently attained by a new parallel implementation for Strassen and ``Strassen-like'' algorithms \cite{BallardDemmelHoltzRomSchwartz11} (see Table~\ref{tbl:lower-upper}).
Interestingly, the numerators here do not depend on $\omega_0$. Thus, an improvement of
$\omega_0$ (the exponent of the arithmetic cost of the algorithm) affects only the power of $p$ in the denominator.

\begin{table}[!h]\label{tbl:lower-upper}
\centering
\small
\begin{tabular}{|l||c|c||c|l|} \hline
 & \multicolumn{2}{c||}{Classical algorithms} & \multicolumn{2}{c|}{``Strassen-like'' algorithms} \\ \cline{2-5}
 & \multicolumn{2}{c||}{ } & \multicolumn{2}{c|}{$2 < \omega_0 < 3$} \\
 \cline{2-5}
 & Lower  bound  & Attained by & Lower  bound  & Attained by \\
 \hline
\hline
Parallel &\cite{IronyToledoTiskin04} &  &  [here]
& \\
$M = $& $\Omega \lt( \lt(\frac{n}{\sqrt{M}}\rt)^3 \cdot \frac{M}{p} \rt)$
&
& $\Omega \lt( \lt(\frac{n}{\sqrt{M}}\rt)^{\omega_0} \cdot \frac{M}{p} \rt)$
& \\
\cline{2-5}
$\;\;\;\; \Theta \lt( \frac{n^2}{p} \rt)$
& $\Omega\lt(\frac{n^2}{p^{\frac12} }  \rt)$
& \cite{Cannon69}
& $\Omega\lt(\frac{n^2}{p^{2-\frac{\omega_0}{2}} }  \rt)$
& [Ballard, \\
&
&
&
& Demmel, \\
\cline{2-4}
$\;\;\;\; \Theta \lt( \frac{n^2}{p^{\frac23}} \rt)$
& $\Omega\lt(\frac{n^2}{p^{\frac23}} \rt)$
& \cite{DekelNassimiSahni81}
& $\Omega\lt( \frac{n^2}{p^{\frac{5-\omega_0}{3}}}\rt)$
& Holtz, \\
&
& \cite{AggarwalChandraSnir90}
&
& Rom, \\
\cline{2-4}
$\;\;\;\; \Theta \lt( c \cdot \frac{n^2}{p} \rt)$
& $\lt( \frac{n^2}{c^{\frac12} p^{\frac12} } \rt)$
& [Solomonik and
& $\Omega\lt( \frac{n^2}{c^{\frac{\omega_0}{2}-1}p^{2-\frac{\omega_0}{2}}}\rt)$
&  Schwartz,\\
$1 \leq c \leq p^{\frac13}$
&
& Demmel 2011]
&
& 2011] \\
\hline
\end{tabular}
\caption{ }
\end{table}

\subsection{Recursive Implementations}\label{sec:recursive-implementations}

In some cases, the simplest recursive implementation of an
algorithm turns out to be communication-optimal (e.g., in the
cases of matrix multiplication
\cite{FrigoLeisersonProkopRamachandran99} and \chold\
\cite{AhmedPingali00,BallardDemmelHoltzSchwartz10a}, but not in
the case of LU decomposition \cite{Toledo97}, which is bandwidth optimal but not latency optimal).
This leads to the
question: when is the communication-optimality of an algorithm
determined by the expansion properties of the corresponding
computation graphs? In this work we showed that such is the case
for ``Strassen-like''
fast matrix multiplication algorithms.

\subsection{Other Hardware}\label{sec:other-hardware} It is of great interest to construct new models general enough to
capture the rich and evolving design space of current and
predicted future computers. Such models can be {\em
homogeneous}, consisting of many layers, where the components of
each layer are the same (e.g., a supercomputer with many identical
multi-core chips on a board, many identical boards in a rack, many
identical racks, and many identical levels of associated memory
hierarchy); or {\em heterogeneous}, with components with different
properties residing on the same level (e.g., CPUs alongside GPUs,
where the latter can do some computations very quickly, but are
much slower to communicate with).

Some experience has been acquired with such systems (see the MAGMA project \cite{MAGMA}, and also \cite{VolkovDemme08} for using GPU assisted linear algebra computation ).
A first step in analyzing such systems has been recently introduced by Ballard, Demmel, and Gearhart
\cite{BallardDemmelGearhart11},
where they modeled heterogenous shared memory architectures, such as mixed GPU/CPU architecture, and obtained tight lower and upper bounds
for $O(n^3)$ matrix multiplication.

Note that we can similarly generalize Corollaries \ref{cor:parallel} and \ref{cor:parallel-str-like}
to other models, such as the heterogenous model
and shared memory model. The reduction is achieved by
observing the communication of a single processor.

However, there is currently no systematic theoretic way of
obtaining upper and lower bounds for arbitrary hardware models.
Expanding such results to other architectures and algorithmic techniques is
a challenging goal.
For example,
recursive algorithms tend to be cache oblivious and communication
optimal for the sequential hierarchy model. Finding an equivalent
technique that would work for an arbitrary architecture is a
fundamental open problem.

\begin{acks}
We thank Eran Rom, Edgar Solomonik, and Chris Umans for helpful discussions.
\end{acks}

\bibliographystyle{acmsmall}
\bibliography{SLB-arXiv}

\appendix

\section{Strassen's Fast Matrix Multiplication Algorithm}\label{sec:strassen-cited}
\noindent Strassen's original algorithm follows \cite{Strassen69}. See \cite{Winograd71} for Winograd's variant, which reduces the number of additions.
\begin{algorithm}
\caption{Matrix Multiplication: Strassen's Algorithm}\label{alg:strassen}
\begin{algorithmic}[1]
\REQUIRE Two $n \times n$ matrices, $A$ and $B$. \IF {$n=1$}
    \STATE  $C_{11} = A_{11}\cdot B_{11}$
\ELSE \STATE \COMMENT{Decompose $A$ into four equal square blocks
$A = \begin{pmatrix}
  A_{11} & A_{12} \\
  A_{21} & A_{22}
\end{pmatrix}$ \\ and the same for $B$.}
\STATE $M_1 = (A_{11} + A_{22}) \cdot (B_{11} + B_{22})$ \STATE
$M_2 = (A_{21} + A_{22})\cdot  B_{11}$ \STATE $M_3 = A_{11}\cdot
(B_{12} - B_{22})$ \STATE $M_4 =        A_{22}\cdot (B_{21} -
B_{11})$ \STATE $M_5 = (A_{11}+ A_{12}) \cdot B_{22}$ \STATE $M_6
= (A_{21} - A_{11}) \cdot (B_{11} + B_{12})$ \STATE $M_7 = (A_{12}
- A_{22}) \cdot (B_{21} + B_{22})$ \STATE $C_{11} = M_1 + M_4 -
M_5 + M_7$ \STATE $C_{12} = M_3 + M_5$ \STATE $C_{21} = M_2 + M_4$
\STATE $C_{22} = M_1 - M_2 + M_3 + M_6$ \ENDIF \RETURN $C$
\end{algorithmic}
\end{algorithm}

\end{document}